\documentclass[3p,11pt,a4paper,twoside,fleqn,sort&compress]{article}
\usepackage{lmodern} 
\usepackage{multicol}
\usepackage[a4paper,margin=1in]{geometry}
\usepackage{amsmath,amssymb,amsthm}

\usepackage{graphicx}
\usepackage{longtable}
\usepackage{enumitem}

\usepackage[utf8]{inputenc}

\usepackage{geometry}
\usepackage{amssymb,amsmath,latexsym}
\usepackage{amsthm}
\usepackage{mathrsfs}
\usepackage[varg]{pxfonts}
\usepackage{multirow}
\newtheorem{theorem}{Theorem}[section]
\newtheorem{proposition}[theorem]{Proposition}
\newtheorem{lemma}[theorem]{Lemma}
\newtheorem{definition}[theorem]{Definition}

\newtheorem{remark}[theorem]{Remark}
\newtheorem{example}[theorem]{Example}

\usepackage{pbsi}
\usepackage[T1]{fontenc}
\usepackage{aurical}
\usepackage[T1]{fontenc}

\usepackage[T1]{fontenc}

\immediate\write18{texcount -tex -sum  \jobname.tex > \jobname.wordcount.tex}

\date{}

\providecommand{\keywords}[1]
{
	\small	
	\textbf{\textit{Keywords---}} #1
}

	\title{Classification of Self-Dual Constacyclic Codes of Prime Power Length $p^s$ Over $\frac{\mathbb{F}_{p^m}[u]}{\left\langle u^3\right\rangle} $}

	\author{ Youssef AHENDOUZ $ ^a $  \thanks{Corresponding author} and Ismail AKHARRAZ $ ^b $
		\\ 
		\small Mathematical and Informatics Engineering Laboratory\\
		\small Ibn Zohr University - Morocco\\
		\small $ ^a $ \texttt{youssef.ahendouz@gmail.com} \qquad$ ^b $ \small \texttt{i.akharraz@uiz.ac.ma }
	}
	
\begin{document}
	
	\maketitle
\begin{abstract} 
	Let $\mathbb{F}_{p^m}$ be a finite field of cardinality $p^m$, where $p$ is a prime number and $m$ is a positive integer. Self-dual constacyclic codes of length \( p^s \) over \( \frac{\mathbb{F}_{p^m}[u]}{\langle u^3 \rangle} \) exist only when \( p = 2 \). In this work, we classify and enumerate all self-dual cyclic codes of length \( 2^s \) over \( \frac{\mathbb{F}_{2^m}[u]}{\langle u^3 \rangle} \), thereby completing the classification and enumeration of self-dual constacyclic codes of length \( p^s \) over \( \frac{\mathbb{F}_{p^m}[u]}{\langle u^3 \rangle} \). Additionally, we correct and improve results from B. Kim and Y. Lee (2020) in \cite{kim2020classification}.
\end{abstract}


	\section{Introduction}

The class of constacyclic codes is fundamental in error-correcting code theory, as it generalizes cyclic codes, a family widely studied since the 1950s. Many important codes, such as BCH, Golay, and Hamming codes, are cyclic or derived from them. Constacyclic codes are especially valued for their efficient encoding using shift registers and effective error-correction capabilities, making them highly useful in various engineering fields

Let \( R \) be a finite commutative ring and \( n \) a positive integer. We consider the \( R \)-module \( R^n = \{(a_0, a_1, \ldots, a_{n-1}) \mid a_i \in R \text{ for } i = 0, 1, \ldots, n-1\} \). The \(\lambda\)-cyclic shift operator \( \sigma_\lambda \) on \( R^n \) is defined by \( \sigma_\lambda(a_0, a_1, \ldots, a_{n-1}) = (\lambda a_{n-1}, a_0, a_1, \ldots, a_{n-2}) \). An \( R \)-submodule \( C \subseteq R^n \) is called \(\sigma_\lambda\)-invariant if \( \sigma_\lambda(C) \subseteq C \); such submodules are known as \(\lambda\)-constacyclic codes, and their elements are called codewords. When \(\lambda = 1\), a \(\lambda\)-constacyclic code \( C \) is referred to as a cyclic code. Each codeword \( a = (a_0, a_1, \ldots, a_{n-1}) \) corresponds to a polynomial \( a(x) = a_0 + a_1 x + \cdots + a_{n-1} x^{n-1} \in \frac{R[x]}{\langle x^n - \lambda \rangle} \), where \( \frac{R[x]}{\langle x^n - \lambda \rangle} = \left\{  \sum\limits_{i=0}^{n-1} a_i x^i \mid a_0, a_1, \ldots, a_{n-1} \in R \right\} \), with arithmetic taken modulo \( x^n - \lambda \). Under this identification, the polynomial \( x a(x) \) corresponds to \( \sigma_\lambda(a_0, a_1, \ldots, a_{n-1})  \).  From this, the following proposition is straightforward.
		
		\begin{proposition}\label{vvvvvvv}\cite[Proposition 2.2]{dinh2010constacyclic}
			The \(\lambda\)-constacyclic codes of length \( n\) over \( R \) are exactly the ideals of \( \frac{R[x]}{\langle x^n - \lambda \rangle} \).
		\end{proposition}

Constacyclic codes over a finite commutative ring \( R \) are classified into two categories: simple-root constacyclic codes, where the codeword length is coprime to the characteristic of \( R \), and repeated-root constacyclic codes, where they are not coprime. The structure of simple-root cyclic and negacyclic codes over finite chain rings has been studied in \cite{dinh2004cyclic}. However, the classification of repeated-root constacyclic codes is limited to specific cases; the general case remains incomplete.

\par For any integer \( t \geq 2 \), let \( R_t = \frac{\mathbb{F}_{p^m}[u]}{\langle u^t \rangle} = \mathbb{F}_{p^m} + u \mathbb{F}_{p^m} + \cdots + u^{t-1} \mathbb{F}_{p^m} \) (with \( u^t = 0 \)).This ring has the maximal ideal \( u R_t \) and was first introduced in \cite{udaya1998optimal} for constructing sequences in frequency hopping multiple-access (FHMA) spreading spectrum communication systems. Since then, it has garnered significant attention and has been the subject of extensive research, particularly in the study of codes over such rings. For example, see \cite{han2011cyclic}.

When \( t = 2 \), the ring \( R_2 \) is commonly used for constacyclic codes. For example, the structure of \( \mathbb{F}_2 + u\mathbb{F}_2 \) is interesting; it can be considered as algebraically lying between \( \mathbb{F}_4 \) and \( \mathbb{Z}_4 \) in the sense that it is additively analogous to \( \mathbb{F}_4 \) and multiplicatively analogous to \( \mathbb{Z}_4 \). Many studies focus on constacyclic codes of various lengths over \( R_2 \) \cite{dinh2010constacyclic, chen2016constacyclic, dinh2017constacyclic, dinh2018negacyclic, dinh2018cyclic, dinh2019class, dinh2019alpha+, ahendouz122negacyclic}. Furthermore, the structures of all self-dual constacyclic codes of length \( p^s \) over \( R_2 \) are given in \cite{dinh2018self}.

\par\par When \( t = 3 \), the units of \( R_3 \) are of the form \( \lambda = \lambda_0 + u \lambda_1 + u^2 \lambda_2 \), where \( \lambda_0, \lambda_1, \lambda_2 \in \mathbb{F}_{p^m} \) and \( \lambda_0 \neq 0 \). If \( \lambda_1 \neq 0 \), \(\lambda\)-constacyclic codes of length \( p^s \) over \( R_3 \), including self-dual codes, have been studied in \cite{dinh2016repeated}. When \( \lambda_1 = 0 \) and \( \lambda_2 \neq 0 \), Sobhani \cite{sobhani2015complete} determined the structure of \((\delta + \alpha u^2)\)-constacyclic codes of length \( p^s \) over \( R_3 \), where \( \delta, \alpha \in \mathbb{F}_{p^m}^\times \), and provided a complete classification of self-dual \((1 + \alpha u^2)\)-constacyclic codes of length \( 2^s \) over \( R_3 \). In the case \( \lambda_1 = \lambda_2 = 0 \), Laaouine et al. studied the structure of all \( \sigma \)-constacyclic codes of length \( p^s \) over \( R_3 \) in \cite{laaouine2021complete} and classified them into \( 8 \) types. While their classification was insightful, it contained errors that we recently corrected in \cite{ahendouz160note}.  The purpose of this paper is to identify all self-dual \( \lambda_0 \)-constacyclic codes of length \( p^s \) over \( R_3 \). According to \cite[Corollary 4.5]{dinh2017repeated} and \cite{sobhani2015complete}, such codes exist only when \( p = 2 \) and \( \lambda_0 = 1 \).

The structure of the paper is as follows: In Section 2, we present some basic definitions, notations, and previous results. Section 3 focuses on the classification of all self-dual cyclic codes of length \( 2^s \) over \( \frac{\mathbb{F}_{2^m}[u]}{\langle u^3 \rangle} \) and their enumeration, improving upon the results obtained in \cite{kim2020classification}.

\section{Preliminaries}

In this section, we present the essential preliminary results on cyclic codes of length \( 2^s \) over the ring \( \frac{\mathbb{F}_{2^m}[u]}{\langle u^3 \rangle} \), where \( m \) and \( s \) are positive integers.  Hereafter, to simplify notation, we denote:
\[
R_3 := \frac{\mathbb{F}_{2^m}[u]}{\left\langle u^3 \right\rangle} = \mathbb{F}_{2^m} + u\mathbb{F}_{2^m} + u^2\mathbb{F}_{2^m}, \quad
\mathcal{R} := \frac{R_3[x]}{\left\langle x^{2^s} + 1 \right\rangle}, \quad \text{and} \quad \mathcal{K} := \frac{\mathbb{F}_{2^m}[x]}{\left\langle x^{2^s} + 1 \right\rangle}.
\]

By Proposition \ref{vvvvvvv}, cyclic codes of length \( 2^s \) over \( R_3 \) correspond to the ideals of \( \mathcal{R} \).  On the other hand, it is well-known that any polynomial \( h(x) \in \mathcal{K} \) can be uniquely expressed as \( h(x) = \sum\limits_{j=0}^{2^s-1} h_j (x+1)^j \), where \( h_j \in \mathbb{F}_{2^m} \). Moreover,  \( h(x) \) is a unit if and only if \( h_0 \neq 0 \). As a consequence, \( \mathcal{K} \) is a finite chain ring with the maximal ideal \( \langle x + 1 \rangle \) and nilpotency index \( 2^s \). This means that its ideals are given by:
\[
\mathcal{K} = \left\langle (x+1)^0 \right\rangle \supsetneq \left\langle (x+1)^1 \right\rangle \supsetneq \cdots \supsetneq \left\langle (x+1)^{2^s-1} \right\rangle \supsetneq \left\langle (x+1)^{2^s} \right \rangle = \{0\}.
\]

Next, we consider the surjective ring homomorphism \( \mu: \mathcal{R} \to \mathcal{K} \), defined by:
\[
\mu(f(x)) = f(x) \bmod u.
\]

\begin{definition}
	Let \( C \) be an ideal of \( \mathcal{R} \). The \( i \)-th torsion code of \( C \), denoted by \( \operatorname{Tor}_i(C) \), is defined as:
	\[
	\operatorname{Tor}_i(C) = \mu\left(\left\{f(x) \in \mathcal{R} \mid u^i f(x) \in C\right\}\right),
	\]
	where \( 0 \leq i \leq 2 \).
\end{definition}

It is clear that \( \operatorname{Tor}_i(C) \) is an ideal of \( \mathcal{K} \). Specifically, for \( 0 \leq i \leq 2 \), we have:
\[
\operatorname{Tor}_i(C) = \left\langle (x+1)^{T_i(C)} \right\rangle,
\]
for some \( 0 \leq T_i(C) \leq 2^s \).

The following theorem provides a classification of cyclic codes of length \( 2^s \) over \( R_3 \) into eight distinct types.

	\begin{theorem}\cite{laaouine2021complete}\label{thm:cyclic_codes}
		Cyclic codes of length \(2^s\) over \(R_3\) i.e., the ideals of \(\mathcal{R}\) are classified into \(8\) types as follows\begin{itemize}
			\item Type \(1:\) \( \langle 0\rangle, \langle 1\rangle\). For this type, we have
			\[
			\begin{aligned}
				& T_0(\langle 0\rangle) = T_1(\langle 0\rangle) = T_2(\langle 0\rangle) = 2^s,\\
				& T_0(\langle 1\rangle) = T_1(\langle 1\rangle) = T_2(\langle 1\rangle) = 0.
			\end{aligned}
			\]
			
			\item Type \(2:\) \( C_2 = \left\langle u^2(x+1)^\tau\right\rangle \), where \(0 \leq \tau \leq 2^s - 1\). For this type,
			\[
			T_0(C_2) = T_1(C_2) = 2^s, \quad T_2(C_2) = \tau.
			\]
			
			\item Type \(3:\) \( C_3 = \left\langle u(x+1)^\delta + u^2(x+1)^t h(x)\right\rangle \), where \(0 \leq L \leq \delta \leq 2^s - 1\), \(0 \leq t < L\), and \(h(x)\) is either \(0\) or a unit in \(\mathcal{K}\), with 
			\begin{equation}\label{L}
				L = \min\left\{k \mid u^2(x+1)^{k} \in C_3\right\} =
				\begin{cases}
					\delta, & \text{if } h(x) = 0, \\ 
					\min\left\{\delta, 2^s - \delta + t\right\}, & \text{if } h(x) \neq 0.
				\end{cases}
			\end{equation}
			For this type, we have
			\[
			T_0(C_3) = 2^s, \quad T_1(C_3) = \delta, \quad T_2(C_3) = L.
			\]
			
			\item Type \(4:\) \( C_4 = \left\langle u(x+1)^\delta + u^2(x+1)^t h(x), u^2(x+1)^\omega\right\rangle \),
			where \(0 \leq \omega < L \leq \delta \leq 2^s - 1\), \(0 \leq t < \omega\), either \(h(x)\) is \(0\) or \(h(x)\) is a unit in \(\mathcal{K}\), and \(L\) defined as in \eqref{L}. For this type, we have
			\[
			T_0(C_4) = 2^s, \quad T_1(C_4) = \delta, \quad T_2(C_4) = \omega.
			\]
			
			\item Type \(5:\) 
			\[
			C_5 = \left\langle (x+1)^a + u(x+1)^{t_1} h_1(x) + u^2(x+1)^{t_2} h_2(x) \right\rangle,
			\]
			where \(1 \leq a \leq 2^s - 1\), \(0 \leq t_1 < U\), \(0 \leq t_2 < V\), with 
			\begin{equation}\label{UU}
				U = \min\left\{ k \mid u(x+1)^{k} + u^2 g(x) \in C_5 \right\} = 
				\begin{cases} 
					a, & \text{if } h_1(x) = 0, \\ 
					\min\left\{ a, 2^s - a + t_1 \right\}, & \text{if } h_1(x) \neq 0,
				\end{cases}
			\end{equation}
			and 
			\begin{equation}\label{VV}
				V = \min\left\{ k \mid u^2(x+1)^{k} \in C_5 \right\}.
			\end{equation}
			
			For this type, we have
			\begin{equation}\label{jdjdj}
				T_0(C_5) = a, \quad T_1(C_5) = U, \quad T_2(C_5) = V.
			\end{equation}
			
			\item Type \(6:\) 
			\[
			C_6 = \left\langle (x+1)^a + u(x+1)^{t_1} h_1(x) + u^2(x+1)^{t_2} h_2(x), u^2(x+1)^c\right\rangle,
			\]
			where \(0 \leq c < V \leq a \leq 2^s - 1\), \(0 \leq t_1 < U\), \(0 \leq t_2 < c\), and for \(i = 1, 2\), \(h_i(x)\) is either \(0\) or a unit in \(\mathcal{K}\). Here, \(U\) and \(V\) are defined as in \eqref{UU} and \eqref{VV}. For this type, we have 
			\[
			T_0(C_6) = a, \quad T_1(C_6) = U, \quad T_2(C_6) = c.
			\]
			
			\item Type \(7:\) 
			\[
			C_7 = \left\langle (x+1)^a + u(x+1)^{t_1} h_1(x) + u^2(x+1)^{t_2} h_2(x), u(x+1)^b + u^2(x+1)^{t_3} h_3(x)\right\rangle,
			\]
			where \(0 \leq b < U \leq a \leq 2^s - 1\), \(0 \leq t_1 < b\), \(0 \leq t_2 < W\), \(0 \leq t_3 < W\), and for \(i = 1, 2, 3\), \(h_i(x)\) is either \(0\) or a unit in \(\mathcal{K}\). Here, \(U\) is defined as in \eqref{UU}, and 
			\begin{equation}\label{WW}
				W = \min\left\{ k \mid u^2(x+1)^{k} \in C_7 \right\}.
			\end{equation}
			For this type, we have 
			\[
			T_0(C_7) = a, \quad T_1(C_7) = b, \quad T_2(C_7) = W.
			\]
			
			\item Type \(8:\) 
			\[
			C_8 = \left\langle (x+1)^a + u(x+1)^{t_1} h_1(x) + u^2(x+1)^{t_2} h_2(x), u(x+1)^b + u^2(x+1)^{t_3} h_3(x), u^2(x+1)^c \right\rangle,
			\]
			where \(0 \leq c < W < U \leq a \leq 2^s - 1\), \(0 \leq t_1 < U\), \(0 \leq t_2 < W\), \(0 \leq t_3 < W\), and for \(i = 1, 2, 3\), \(h_i(x)\) is either \(0\) or a unit in \(\mathcal{K}\). Here, \(U\) and \(W\) are defined as in \eqref{UU} and \eqref{WW}. For this type, we have 
			\[
			T_0(C_8) = a, \quad T_1(C_8) = U, \quad T_2(C_8) = W.
			\]
		\end{itemize}
	\end{theorem}
		\begin{theorem}\cite[Theorem 2]{ahendouz160note} Let $V$ be as defined in \eqref{VV}, then we have
	
	\[		V =
			\begin{cases}
				a, & \text{if } h_1(x) = h_2(x) = 0, \\
				\min\{ a,2^{s} - a + t_{2} \}, & \text{if } h_1(x) = 0 \text{ and } h_2(x) \neq 0, \\
				\min\{ a,2^{s}-2 a+  2t_{1}\}	 ,& \text{if } h_{1}(x) \neq 0,h_2(x)=0 \text{ and } a \leq 2^s - a + t_{1} ,\\
				\min \{a,2^s- a+  t_{2},2^{s}-2 a+  2t_{1}\} ,& \text{if } h_{1}(x)\neq0, h_2(x)\neq 0, a \leq 2^s - a + t_{1} \text{ and }  2t_{1}\neq a  +  t_{2},\\
				\min\{ a,2^s- a+  t_{2}+\alpha_1\}	,& \text{if } h_{1}(x) \neq 0,h_2(x)\neq 0, a \leq 2^s - a + t_{1}  \text{ and }  2t_{1}= a  +  t_{2},
				\\
				t_1,& \text{if } h_1(x)\neq 0, h_2(x)=0\text{ and } a \geq 2^s - a + t_{1},\\
				\min \{t_{1},a+  t_{2}-  t_{1} \},& \text{if } h_1(x)\neq0, h_2(x)\neq 0,a \geq 2^s - a + t_{1}  \text{ and }  2t_{1}\neq a  +  t_{2},\\ \min\left\{ 2^s+t_1-a,	t_{1}+\alpha_1\right\}
				,&  \text{if } h_1(x)\neq0, h_2(x)\neq 0,a \geq 2^s - a + t_{1}  \text{ and }  2t_{1}= a  +  t_{2},
			\end{cases}\]
		where 
	\[	\alpha_1=\max\left\{0\leq k\leq 2^s\mid  (x+1)^{k} \text{ divides } 	h_1(x) - h_2(x) h_1(x)^{-1} \right\}.\]

	\end{theorem}
\begin{theorem}\cite[Theorem 3]{ahendouz160note}\label{rrfrfr}
	Let $W$ be defined as in \eqref{WW}. Then we have:
	
\begin{eqnarray}\label{zvdeefefefcefcefefc}
		W = \begin{cases}
		b, & \text{if }h_1(x) =  h_{2}(x) = h_{3}(x) = 0, \\
		\min\{ b, a- b+  t_{3}\}, & \text{if } h_1(x) = h_{2}(x)=0 \text{ and } h_{3}(x)\neq 0, \\
		\min\{2^{s}- a+  t_{2}, b\}, & \text{if } h_{2}(x)\neq0 \text{ and }h_1(x) =  h_{3}(x)= 0, \\
		\min\{2^{s}- a+  t_{2}, b, a- b+  t_{3}\}, & \text{if } h_1(x)=0 ,h_{2}(x)\neq 0 \text{ and } h_{3}(x)\neq 0,\\
		t_1, & \text{if } h_1(x)\neq 0 \text{ and } h_2(x)=h_3(x)=0, \\
		\min\{2^{s}- a+  t_{2},t_1 \} ,& \text{if } h_1(x)\neq0, h_2(x)\neq0 \text{ and } h_3(x)=0,\\
		\min\{\beta_3,\beta_4, b,2^{s}- b+  t_{3}\},&\text{if } h_1(x)\neq0 \text{ and }h_3(x)\neq0,
	\end{cases}
\end{eqnarray}	where 
	\[\left.\begin{array}{rcl}
		\beta_3 &=& \begin{cases}
			2^{s}- a + t_{1}- b+ t_{3} ,& \text{if } h_{2}(x)= 0, \\
			\min\{ 2^{s}- a + t_{2}, 2^{s}- a + t_{1}- b+ t_{3} \}, & \text{if } h_{2}(x)\neq0 \text{ and }    t_{2}  \neq \  t_{1}- b+  t_{3}, \\
			2^{s}- a + t_{1}- b+ t_{3} +\alpha_{3}, & \text{if } h_{2}(x)\neq0 \text{ and }   t_{2}  = \  t_{1}- b+  t_{3},
		\end{cases} \\
		\beta_4 &=& \begin{cases}
			\min\{  t_{1}, a- b+  t_{3} \} ,& \text{if }   t_{1}  \neq  a- b+  t_{3}, \\
			t_{1}+\alpha_{4} ,& \text{if }    t_{1}  =  a- b+  t_{3},
		\end{cases}\\ 	\alpha_3 &=& \max \left\{   0 \leq k \leq 2^s \mid (x+1)^k \text{ divides }    h_2(x)- h_{1}(x) h_{3}(x) \right\},\\	\alpha_4 &=& \max \left\{ 0 \leq k \leq 2^s \mid (x+1)^k \text{ divides }   h_1(x)-h_{3}(x) \right\}.
	\end{array}\right.\]

\end{theorem}
	
Given two polynomials \( c(x) = \sum\limits_{j=0}^{2^s-1} c_j x^j \) and \( d(x) = \sum\limits_{j=0}^{2^s-1} d_j x^j \) in \( \mathcal{R} \), the dot product of \( c(x) \) and \( d(x) \) is defined as follows \[
c(x) \cdot d(x) = \sum\limits_{j=0}^{2^s-1} c_j d_j.
\]

With respect to this dot product, the dual of a cyclic code \( C \) of length \( 2^s \) over \( R_3 \) (i.e., an ideal in \( \mathcal{R} \)), denoted \( C^{\perp} \), is given by
\[
C^{\perp} = \left\{ c(x) \in \mathcal{R} \mid c(x) \cdot d(x) = 0 \text{ for all } d(x) \in C \right\}.
\]

A cyclic code \( C \) is said to be self-dual if \( C = C^{\perp} \). It is well known that the dual code \( C^{\perp} \) is also a cyclic code of length \( 2^s \) over \( R_3 \), and hence an ideal of \( \mathcal{R} \).

	\begin{definition}
		Let \( C \) be a cyclic code of length \( 2^s \) over \(  R_3  \).  The annihilator of \( C \), denoted \( \mathcal{A}(C) \), is the set
		\[
		\mathcal{A}(C) = \left\{ c(x) \in \mathcal{R} \mid \forall d(x) \in C, \,c(x) d(x)  = 0 \right\}.
		\]
	\end{definition}
	
	Consider the polynomial \( c(x) = \sum\limits_{j=0}^{2^s-1} c_j x^j \in \mathcal{R} \). The reciprocal of \( c(x) \), denoted \( c^*(x) \), is defined as $ x^lc(x^{-1}) $ where $ l=\deg c(x). $
	\begin{lemma}\cite{sobhani2015complete}\label{ccecececcv}
		Let \(C\) be a cyclic code of length \(2^s\) over \( R_3 \). Then \[C^{\perp} = \mathcal{A}(C)^* :=\{c^*(x)\mid c(x)\in\mathcal{A}(C)\}.\]
	\end{lemma}

	\begin{lemma}[\cite{sobhani2015complete}]
		Let \( C \) be a cyclic code of length \( 2^s \) over \( R_3 \). Then, for all \( 0 \leq i \leq 2 \), we have
		\[
		T_i(C^\perp) = T_i(\mathcal{A}(C)) = 2^s - T_{2-i}(C).
		\]
	\end{lemma}
	
	As a consequence of the previous lemma, for any self-dual cyclic code of length \( 2^s \) over \( R_3 \), the following holds
	\begin{equation} \label{jececececececec}
		T_2(C) = 2^s - T_0(C) \quad \text{and} \quad T_1(C) = 2^{s-1}.
	\end{equation}

\section{Self-Dual Cyclic Codes of Length \(2^s\) over \(R_3\)}

In this section, we classify all self-dual cyclic codes of length \( 2^s \) over \( R_3 \). From Eq. \eqref{jececececececec}, it follows that every self-dual cyclic code \( C \) must satisfy \( T_1(C) = 2^{s-1} \). Consequently, there are no self-dual codes of types 1 and 2.

For cyclic codes of type 3, assuming \( C \) is self-dual, Eq. \eqref{jececececececec} implies \( \delta = 2^{s-1} \) and \( L = 0 \), contradicting Eq. \eqref{L}. Hence, such codes cannot be self-dual.

The following theorem provides the necessary and sufficient conditions for the self-duality of cyclic codes of type 4.
\begin{theorem}
	Let \( C = \left\langle u(x+1)^\delta + u^2(x+1)^t h(x), u^2(x+1)^\omega \right\rangle \) be a cyclic code of length \( 2^s \) over \( R_3 \) of type $ 4. $ Then, \( C \) is self-dual if and only if \( C = \left\langle u(x+1)^{2^{s-1}}, u^2 \right\rangle \).
\end{theorem}

\begin{proof}  The torsion codes of \(C\) satisfy \(T_0(C) = 2^s\), \(T_1(C) = \delta\), and \(T_2(C) = \omega\). Substituting these values into Eq.~\eqref{jececececececec}, we find \(\delta = 2^{s-1}\) and \(\omega = 0\). Consequently, \(C\) reduces to:
	\[
	C = \left\langle u(x+1)^{2^{s-1}}, u^2 \right\rangle.
	\]
	
 In this case, it follows that \( \mathcal{A}(C) = \left\langle u(x+1)^{2^{s-1}}, u^2 \right\rangle \), and thus \( C^\perp = \left\langle u(x+1)^{2^{s-1}}, u^2 \right\rangle \), which establishes that \( C \) is indeed self-dual.
\end{proof}

We now proceed the classification to cyclic codes of types 5, 6, 7, and 8. Any such code \(C\) can be expressed in the general form:
\[
C = \left\langle (x+1)^a + u(x+1)^{t_1} h_1(x) + u^2(x+1)^{t_2} h_2(x), \, u(x+1)^b + u^2(x+1)^{t_3} h_3(x), \, u^2(x+1)^c \right\rangle,
\]
where \( b \leq U \) and \( c \leq W \). For a proof, refer to that of \cite[Theorem 1]{laaouine2021complete}. Additionally, 
\[
T_1(C) = a, \quad T_2(C) = b, \quad \text{and} \quad T_3(C) = c.
\]

From Eq.~\eqref{jececececececec}, it follows that \( b = 2^{s-1} \) and \( c = 2^s - a \).  Substituting these values, we obtain the equivalent representation:

\begin{equation} \label{7}
	\begin{aligned}
		C = &\left\langle (x+1)^a + u(x+1)^{t_1} h_1(x) + u^2(x+1)^{t_2} h_2(x) , u(x+1)^{2^{s-1}} + u^2(x+1)^{t_3} h_3(x), \, u^2(x+1)^{2^s - a} \right\rangle,
	\end{aligned}
\end{equation}
with \( 2^{s-1} \leq U \) and \( 2^s - a \leq W \). By Eq.~\eqref{UU}, \( 2^{s-1} \leq U \) is equivalent to:
\begin{eqnarray}\label{zojeededeD}
	\begin{cases} 
		2^{s-1} \leq a, & \text{if } h_1(x) = 0, \\ 
		2^{s-1} \leq a \leq 2^{s-1} + t_1, & \text{if } h_1(x) \neq 0.
	\end{cases}
\end{eqnarray}

\begin{remark}
The condition \( 2^s - a \leq W \) is always satisfied when \( C = C^\perp \). Assume, for contradiction, that \( 2^s - a > W \). In this case, we would have \( u^2(x+1)^{2^s - a - 1} \in C = C^\perp \), which leads to the equation
\[
(x+1)^{2^s - 1} = u^2(x+1)^{2^s - a - 1} \left( (x+1)^a + u(x+1)^{t_1} h_1(x) + u^2(x+1)^{t_2} h_2(x) \right) = 0,
\]
which is a contradiction. Hence, this condition can be omitted from further verification.

\end{remark}

The following pair of lemmas plays a crucial role in the classification of self-dual codes defined in Eq.~\eqref{7} in the remainder of this section.

\begin{lemma}\label{kldzjkejee}
	Let \( C \) be the cyclic code defined in \eqref{7}.
	 Suppose there exist \( f_1(x), f_2(x), f_3(x) \in \mathcal{K} \) such that
 \[
(x+1)^a + u f_1(x) + u^2 f_2(x) \in \mathcal{A}(C) \quad \text{and} \quad
u(x+1)^{2^{s-1}} + u^2 f_3(x) \in \mathcal{A}(C).
\]
	Then, the annihilator of \( C \) is given by
	\[
	\mathcal{A}(C) = \left\langle (x+1)^a + u f_1(x) + u^2 f_2(x), \; u(x+1)^{2^{s-1}} + u^2 f_3(x), \; u^2(x+1)^{2^s - a} \right\rangle.
	\]
\end{lemma}

\begin{proof}
	It is straightforward to verify that
	\[
	\left\langle (x+1)^a + u f_1(x) + u^2 f_2(x), \; u(x+1)^{2^{s-1}} + u^2 f_3(x), \; u^2(x+1)^{2^s - a} \right\rangle \subseteq \mathcal{A}(C).
	\]
	
	Conversely, let \( a(x) \in \mathcal{A}(C) \). We have, \( \mu(a(x)) \in \operatorname{Tor}_0(\mathcal{A}(C)) = \left\langle (x+1)^a \right\rangle \). Thus, we can write $ 	\mu(a(x)) = l_1(x)(x+1)^a, $
	where \( l_1(x) \in \mathcal{K} \). It follows that
	$ a(x) - l_1(x) \left( (x+1)^a + u f_1(x) + u^2 f_2(x) \right) \in \mathcal{A}(C) \cap \langle u \rangle. $	This implies
$ 	a(x) - l_1(x) \left( (x+1)^a + u f_1(x) + u^2 f_2(x) \right) = u \left( t_1(x) + u t_2(x) \right), $
	where \( t_1(x), t_2(x) \in \mathcal{K} \). Since \( t_1(x) \in \operatorname{Tor}_1(\mathcal{A}(C)) = \left\langle (x+1)^{2^{s-1}} \right\rangle \), we can express	$ t_1(x) = l_2(x)(x+1)^{2^{s-1}}, $	for some \( l_2(x) \in \mathcal{K} \). Therefore
	\[
	a(x) - l_1(x) \left( (x+1)^a + u f_1(x) + u^2 f_2(x) \right) - l_2(x) \left( u(x+1)^{2^{s-1}} + u^2 f_3(x) \right) \in \mathcal{A}(C) \cap \langle u^2 \rangle.
	\]
	
Using the same reasoning as above, there exists \( l_3(x) \in \mathcal{K} \) such that
\[
a(x) - l_1(x) \left( (x+1)^a + u f_1(x) + u^2 f_2(x) \right) - l_2(x) \left( (x+1)^{2^{s-1}} + u^2 f_3(x) \right) = l_3(x)(x+1)^{2^s - a}.
\]

Hence, we conclude
\[
\mathcal{A}(C) = \left\langle (x+1)^a + u f_1(x) + u^2 f_2(x), \; u(x+1)^{2^{s-1}} + u^2 f_3(x), \; u^2(x+1)^{2^s - a} \right\rangle.
\]
\end{proof}

\begin{lemma}\label{dcedcecec}
Let \( k \geq r \), and \( h_j \in \mathbb{F}_{2^m} \). The following identity holds
	\[
	\sum\limits_{j=0}^{r-1} h_j (x+1)^j x^{k-j} = \sum\limits_{\ell=0}^{r-1} \sum\limits_{j=0}^{\ell} h_j \binom{k-j}{\ell-j} (x+1)^\ell \pmod{ (x+1)^{r}}.
	\]
\end{lemma}
\begin{proof}
	We start with the left-hand side
	\[
	\begin{aligned}
		\sum\limits_{j=0}^{r-1} h_j (x+1)^j x^{k-j} 
		& = \sum\limits_{j=0}^{r-1} h_j (x+1)^j [(x+1) + 1]^{k-j} \\
		& = \sum\limits_{j=0}^{r-1} h_j (x+1)^j \sum\limits_{\ell=0}^{k-j} \binom{k-j}{\ell} (x+1)^{\ell} \\
		& = \sum\limits_{j=0}^{r-1}  \sum\limits_{\ell=j}^{k} h_j \binom{k-j}{\ell-j} (x+1)^\ell\\
		& = \sum\limits_{j=0}^{r-1}  \sum\limits_{\ell=j}^{r-1} h_j \binom{k-j}{\ell-j} (x+1)^\ell \pmod{(x+1)^{r}}\\
		& = \sum\limits_{\ell=0}^{r-1}  \sum\limits_{j=0}^{\ell} h_j \binom{k-j}{\ell-j} (x+1)^\ell \pmod {(x+1)^{r}}.
	\end{aligned}
	\]

	Thus, the identity is proved.
\end{proof}

The following theorem establishes the necessary and sufficient conditions for the cyclic code defined in \eqref{7} to be self-dual when \( h_1(x) = 0 \).

\begin{theorem}\label{11111}
	Let \( C \) be the cyclic code defined in \eqref{7} with \( h_1(x) = 0 \). Then \( C \) is self-dual if and only if  
	\[
	C = \left\langle (x+1)^a + u^2 h(x), u(x+1)^{2^{s-1}}, u^2(x+1)^{2^s-a} \right\rangle,
	\]
	where \( 2^{s-1} \leq a \leq 2^s-1 \), \( h(x) \in \mathcal{K} \) with \( \operatorname{deg}(h(x)) < 2^s-a\), and the coefficients of \( h(x) \) satisfy the equation:
	\[
	M(a)\left(h_0, h_1, \ldots, h_{2^s-a-1}\right)^\mathrm{tr} =
	(0, 0, \ldots, 0)^\mathrm{tr},
	\]
	with
	\[
	M(a) =
	\begin{pmatrix}
		0 & 0 & 0 & \cdots & 0 \\
		\binom{a}{1} & 0 & 0 & \cdots & 0 \\
		\binom{a}{2} & \binom{a-1}{1} & 0 & \cdots & 0 \\
		\vdots & \vdots & \vdots & \ddots & \vdots \\
		\binom{a}{2^s-a-1} & \binom{a-1}{2^s-a-2} & \binom{a-2}{2^s-a-3} & \cdots & 0
	\end{pmatrix}.
	\]

\end{theorem}

\begin{proof}
	From Eq.~\eqref{zojeededeD}, we know that \( 2^{s-1} \leq a \leq 2^s - 1 \). Next, we show that \( h_3(x) = 0 \). To prove this, note that \( u(x+1)^{2^{s-1}} \in \mathcal{A}(C) \). By Lemma~\ref{ccecececcv}, it follows that:
	\[
	u(x+1)^{2^{s-1}} = \left( u(x+1)^{2^{s-1}} \right)^* \in C^\perp.
	\]
	
If \( C \) is self-dual, then \( u(x+1)^{2^{s-1}} \in C \). Consequently, \( u^2(x+1)^{t_3} h_3(x) \in C \). Furthermore, in cases where \( h_3(x) \) is a unit, it follows that \( u^2(x+1)^{t_3} \in C \), which leads to a contradiction since \( t_3 < 2^s-a \) and $\operatorname{Tor}_2(C)= \left\langle (x+1)^{2^s-a} \right\rangle $. Thus, \( h_3(x) = 0 \).

	Now consider the polynomials \( (x+1)^a + u^2 (x+1)^{t_2} h_2(x) \), \( u(x+1)^{2^{s-1}} \), and \( u^2(x+1)^{2^s-a} \). Clearly, these polynomials belong to \( \mathcal{A}(C) \). Applying Lemma~\ref{kldzjkejee}, we obtain:
	\[
	\mathcal{A}(C) = \left\langle (x+1)^a + u^2 (x+1)^{t_2} h_2(x), \ u(x+1)^{2^{s-1}}, \ u^2(x+1)^{2^s-a} \right\rangle.
	\]
	
	Let \( h(x) = (x+1)^{t_2} h_2(x) = \sum\limits_{j=0}^{2^s-a-1} h_j (x+1)^j \). Using Lemmas~\ref{ccecececcv} and \ref{dcedcecec}, the dual code \( C^\perp \) can be expressed as:
	\[
	\begin{aligned}
		C^\perp &= \left\langle (x+1)^a + u^2 \sum\limits_{j=0}^{2^s-a-1} h_j (x+1)^j x^{a-j}, \ u(x+1)^{2^{s-1}}, \ u^2(x+1)^{2^s-a} \right\rangle \\
		&= \left\langle (x+1)^a + u^2 \sum\limits_{\ell=0}^{2^s-a-1} \left(\sum\limits_{j=0}^\ell h_j \binom{a-j}{\ell-j}\right) (x+1)^\ell, \ u(x+1)^{2^{s-1}}, \ u^2(x+1)^{2^s-a} \right\rangle.
	\end{aligned}
	\]
	
	Thus, \( C \) is self-dual if and only if:
	\begin{equation}\label{jckjezcedcdec}
		h_\ell = \sum\limits_{j=0}^\ell h_j \binom{a-j}{\ell-j}, \quad \text{for} \quad \ell = 0, 1, \dots, 2^s-a-1.
	\end{equation}

	This equation can be rewritten as: 
	\begin{eqnarray}\label{bbbb}
		M(a)\left(h_0, h_1, \ldots, h_{2^s-a-1}\right)^\mathrm{tr} =
		(0, 0, \ldots, 0)^\mathrm{tr},
	\end{eqnarray}
	where \( M(a) \) denotes the matrix defined in the theorem.
\end{proof}

We now consider the case where \( h_1(x) \) is a unit. The following lemma lists the necessary conditions for the self-dual cyclic code defined in \eqref{7}.

\begin{lemma}\label{jezoece}
	Let \( C \) be the cyclic code defined in \eqref{7}, where \( h_1(x) \) is a unit. If \( C \) is self-dual, then we must have the following:
\begin{enumerate}[label=\((\mathrm{C}_{\arabic*})\)]
		\item \label{i0I} \( 2^{s-1} \leq a \leq 2^{s-1} + t_1 \).
		\item \label{iI} \( h_2(x) \) is a unit, and \( 2t_1 \geq a + t_2 \).
		\item \label{iII} \( (x+1)^{t_3} h_3(x) = (x+1)^{2^{s-1} + t_1 - a} x^{a - t_1} h_1(x^{-1}). \)
		\item \label{iIII} \( h_1(x) = x^{a-t_1} h_1(x^{-1}) \pmod{(x+1)^{2^{s-1}-t_1}}. \)
	\end{enumerate}
\end{lemma}

\begin{proof} The condition \ref{i0I}  immediately follow from Eq.~\eqref{zojeededeD}.
	 Next
	From Lemma~\ref{ccecececcv} and the self-duality of \( C \), it follows that
	\[
	\begin{aligned}
		&(x+1)^a + u(x+1)^{t_1} x^{a-t_1} h_1(x^{-1}) + u^2(x+1)^{t_2} x^{a-t_2} h_2(x^{-1}) \\
		&= \left( (x+1)^a + u(x+1)^{t_1} h_1(x) + u^2(x+1)^{t_2} h_2(x) \right)^* \in C^* = \mathcal{A}(C).
	\end{aligned}
	\]
	This leads to the following two  equations.
	\[
	\begin{aligned}
		&\left( (x+1)^a + u(x+1)^{t_1} x^{a-t_1} h_1(x^{-1}) + u^2(x+1)^{t_2} x^{a-t_2} h_2(x^{-1}) \right) \cdot \\
		&\quad \left( (x+1)^a + u(x+1)^{t_1} h_1(x) + u^2(x+1)^{t_2} h_2(x) \right) = 0, \\
		&\left( (x+1)^a + u(x+1)^{t_1} x^{a-t_1} h_1(x^{-1}) + u^2(x+1)^{t_2} x^{a-t_2} h_2(x^{-1}) \right) \cdot \\
		&\quad \left( u(x+1)^{2^{s-1}} + u^2(x+1)^{t_3} h_3(x) \right) = 0.
	\end{aligned}
	\]
	By isolating terms involving \( u^2 \), we deduce the following equations:
	\begin{align}
		(x+1)^{a+t_2} h_2(x) &= (x+1)^{2t_1} x^{a-t_1} h_1(x) h_1(x^{-1}) + (x+1)^{a+t_2} x^{a-t_2} h_2(x^{-1}), \label{eq:h2}\\
		(x+1)^{a+t_3} h_3(x) &= (x+1)^{2^{s-1}+t_1} x^{a-t_1} h_1(x^{-1}). \label{eq:h3}
	\end{align}
	
	Assume, for contradiction, that \( h_2(x) = 0 \), then \( h_2(x^{-1}) = 0 \), and substituting into Eq.~\eqref{eq:h2} gives.
	\[
	(x+1)^{2t_1} x^{a-t_1} h_1(x) h_1(x^{-1}) = 0,
	\]
	which is impossible since \( h_1(x) \) is a unit and \( 2t_1 < 2^s \). Thus, \( h_2(x) \) must be a unit. Using Eq.~\eqref{eq:h2} again we find:
	\[
	(x+1)^{2t_1} x^{a-t_1} h_1(x) h_1(x^{-1}) \equiv 0 \pmod{(x+1)^{a+t_2}},
	\]
	implying \( 2t_1 \geq a + t_2 \), which establishes condition~\ref{iI}.	Next	from Eq.~\eqref{eq:h3}, we deduce
	\[
	(x+1)^{t_3} h_3(x) = (x+1)^{2^{s-1} + t_1 - a} x^{a-t_1} h_1(x^{-1}) \pmod{(x+1)^{2^s - a}}.
	\]
	
	Since \( (x+1)^{t_3} h_3(x) \) is defined modulo \( (x+1)^{2^s - a} \), this establishes condition~\ref{iII}.	Finally, substituting condition~\ref{iII} into Eq.~\eqref{eq:h3} yields
	\[
	(x+1)^{2^{s-1} + t_1} h_1(x) = (x+1)^{2^{s-1} + t_1} x^{a-t_1} h_1(x^{-1}),
	\]
	which is equivalent to condition~\ref{iIII}. This completes the proof.
\end{proof}

Let us update the form of the code \( C \) using condition \ref{iII} of Lemma \ref{jezoece}, which gives the following:
\begin{equation} \label{8}
	\begin{aligned}
		C = &\left\langle (x+1)^a + u(x+1)^{t_1} h_1(x) + u^2(x+1)^{t_2} h_2(x), \right. \\ 
		&\left. u(x+1)^{2^{s-1}} + u^2(x+1)^{2^{s-1} + t_1 - a} x^{a - t_1} h_1(x^{-1}), \, u^2(x+1)^{2^s - a} \right\rangle.
	\end{aligned}
\end{equation}

Next, let \( D \) be the ideal of \(\mathcal{R}\) defined as
\[
\begin{aligned}
	D = \left\langle (x+1)^{a} + u(x+1)^{t_1} h_1(x) + u^2 (x+1)^{t_2} F(x), \, u(x+1)^{2^{s-1}} + u^2(x+1)^{2^{s-1} + t_1 - a} h_1(x), \, u^2(x+1)^{2^{s}-a} \right\rangle,
\end{aligned}
\]
where \[ F(x) = h_2(x) + (x+1)^{2t_1 - a - t_2} h_1(x)^2 .\] 

Assuming that the conditions of Lemma \ref{jezoece} hold, a straightforward calculation shows that \( D \subseteq \mathcal{A}(C) \). By Lemma \ref{kldzjkejee}, it follows that \( \mathcal{A}(C) = D \). Consequently, using condition \ref{iIII} of Lemma \ref{jezoece}, the dual code \( C^\perp \) can be expressed as
\[
\begin{aligned}
	C^\perp = D^* = & \left\langle (x+1)^{a} + u(x+1)^{t_1} x^{a - t_1} h_1(x^{-1}) + u^2 (x+1)^{t_2} x^{a - t_2} F(x^{-1}), \right. \\ 
	& \qquad\qquad \left. u(x+1)^{2^{s-1}} + u^2(x+1)^{2^{s-1} + t_1 - a} x^{a-t_1}h_1(x^{-1}), \, u^2(x+1)^{2^{s}-a} \right\rangle \\
	= & \left\langle (x+1)^{a} + u(x+1)^{t_1} h_1(x) + u^2 (x+1)^{t_2} x^{a - t_2} F(x^{-1}), \right. \\ 
	& \qquad\qquad \left. u(x+1)^{2^{s-1}} + u^2(x+1)^{2^{s-1} + t_1 - a} x^{a-t_1}h_1(x), \, u^2(x+1)^{2^{s}-a} \right\rangle.
\end{aligned}
\]

Thus, \( C^\perp = C \) if and only if 
\[
x^{a - t_2} F(x^{-1}) \equiv h_2(x) \pmod{(x+1)^{2^s - t_2-a}}.
\]

By substituting \( x \) with \( x^{-1} \), this condition is equivalent to:
\begin{enumerate}[label=\((\mathrm{C}_{\arabic*})\), start=5]
	\item\label{item:IIIII} 
	\( h_2(x) = x^{a - t_2} h_2(x^{-1}) + (x + 1)^{2t_1 - a - t_2} x^{a - t_1} h_1(x)^2 \pmod{(x + 1)^{2^s - t_2 - a}}. \)
\end{enumerate}

Let \( h_1(x) = \sum\limits_{j=0}^{2^{s-1} - t_1 - 1} a_j (x + 1)^j \), \( h_2(x) = \sum\limits_{j=0}^{2^s - a - t_2 - 1} b_j (x + 1)^j \), and 
\[
(x + 1)^{2t_1 - a - t_2} x^{a - t_1} h_1(x)^2 = \sum\limits_{j=0}^{2^s-a-t_2-1} c_j (x + 1)^j \pmod{(x + 1)^{2^s-t_2-a}},
\]
where \( a_j, b_j, c_j \in \mathbb{F}_{2^m} \). Consider also the following matrices:

\[
N(a, t) =
\begin{pmatrix}
	0 & 0 & 0 & \cdots & 0 \\
	\binom{a-t}{1} & 0 & 0 & \cdots & 0 \\
	\binom{a-t}{2} & \binom{a-t-1}{1} & 0 & \cdots & 0 \\
	\vdots & \vdots & \vdots & \ddots & \vdots \\
	\binom{a-t}{2^{s-1}-t-1} & \binom{a-t-1}{2^{s-1}-t-2} & \binom{a-t-2}{2^{s-1}-t-3} & \cdots & 0
\end{pmatrix},
\]

\[
K(a, t) =
\begin{pmatrix}
	0 & 0 & 0 & \cdots & 0 \\
	\binom{a-t}{1} & 0 & 0 & \cdots & 0 \\
	\binom{a-t}{2} & \binom{a-t-1}{1} & 0 & \cdots & 0 \\
	\vdots & \vdots & \vdots & \ddots & \vdots \\
	\binom{a-t}{2^s-a-t-1} & \binom{a-t-1}{2^s-a-t-2} & \binom{a-t-2}{2^s-a-t-3} & \cdots & 0
\end{pmatrix}.
\]

Using Lemma~\ref{dcedcecec}, condition \ref{iIII} can be rewritten as:
\[
\sum\limits_{\ell=0}^{2^{s-1}-t_1 - 1} a_\ell (x + 1)^\ell = \sum\limits_{\ell=0}^{2^{s-1}-t_1 - 1} \sum\limits_{j=0}^\ell a_j \binom{a - t_1}{\ell - j} (x + 1)^\ell,
\]
which simplifies in matrix form to:
\[
N(a, t_1)(a_0, a_1, \ldots, a_{2^{s-1}-t_1-1})^{\mathrm{tr}} = \mathbf{0}.
\]

Similarly, condition \ref{item:IIIII} is equivalent to:
\[
K(a, t_2)(b_0, b_1, \ldots, b_{2^s-a-t_2-1})^{\mathrm{tr}} = \left(c_0, c_1, \ldots, c_{2^s-a-t_2-1}\right)^{\mathrm{tr}}.
\]

Next, observe that the \((1,1)\)-entry of \( K(a, t_2) \) is null, which implies \( c_0 = 0 \). Consequently, \( 2t_1 > a + t_2 \). Finally, from conditions \ref{i0I} and \ref{iI}, it follows in particular that \( s \geq 2 \).

In light of the above discussion and the preceding notation, we arrive at the following conclusion:

\begin{theorem} \label{22222}
	Let \( C \) be the self-dual cyclic code defined in \eqref{7}, where \( h_1(x) \) is a unit. Then \( C \) is self-dual if and only if  
	\[
	\begin{aligned}
		C =& \left\langle (x+1)^a + u(x+1)^{t_1} h_1(x) + u^2(x+1)^{t_2} h_2(x), \right.\\
		&\left. u(x+1)^{2^{s-1}} + u^2(x+1)^{2^{s-1} + t_1 - a} x^{a - t_1 - \deg h_1(x)} h^*_1(x), \, u^2(x+1)^{2^s - a} \right\rangle,
	\end{aligned}
	\]
	where the following conditions hold:
	\begin{enumerate}[label=\Roman*)]
		\item \( s \geq 2, \)
		\item \label{dzcedcecec} \( 2^{s-1} \leq a \leq 2^{s-1} + t_1 \), \( 2t_1 > a + t_2 \),
		\item \label{dzcedcececa} \( h_1(x) \) and \( h_2(x) \) are units such that:
		\begin{eqnarray}\label{ccczcz}
			N(a, t_1)(a_0, a_1, \ldots, a_{2^{s-1}-t_1-1})^{\mathrm{tr}} = \mathbf{0},
		\end{eqnarray}
		and
		\begin{eqnarray}\label{ccczcz2}
			K(a, t_2)(b_0,b_1, \ldots, b_{2^s-a-t_2-1})^{\mathrm{tr}} = (c_0,c_1, \ldots, c_{2^s-a-t_2-1})^{\mathrm{tr}}.
		\end{eqnarray}
	\end{enumerate}
\end{theorem}
\begin{remark} Let \( C \) denote the self-dual cyclic code given in Theorem~\ref{11111}. From Eqs.~\eqref{UU} and \eqref{zvdeefefefcefcefefc}, the integers \( U \) and \( W \) are determined as \( U = a \) and \( W = \min\{2^s - a + t, 2^{s-1}\} \), where \( (x+1)^t = \gcd\big((x+1)^{2^s}, h(x)\big) \). It follows that \( C \) is of type \( 5 \) if and only if \( a = 2^{s-1} \). In the case \( a > 2^{s-1} \), \( C \) is of type \( 7 \) if \( h(x) \) is a unit; otherwise, \( C \) is of type \( 8 \).  Similarly, for the self-dual cyclic code \( C \) given in Theorem~\ref{22222}, we have \( U = a \) and \( W = \min\{2^s - a + t_2, 2^{s-1}\} \). In this case, \( C \) is of type \( 5 \) if and only if \( a = 2^{s-1} \), of type \( 7 \) when \( a > 2^{s-1} \) and \( t_2 = 0 \), and of type \( 8 \) otherwise.
\end{remark}

We now proceed to determine the number of self-dual cyclic codes of length \( 2^s \) over the ring \( R_3 \).  To do this, we use results from \cite{Kiah2012}. For positive integers \( a \) and \( b \) with \( a \geq b \), the matrix \( T(a+b, b) \) is given by:
\begin{eqnarray}\label{matr}
	T(a+b, b)=\left(\begin{array}{ccccc}
		(-1)^a+1 & 0 & 0 & \cdots & 0 \\
		(-1)^a\binom{a}{1} & (-1)^{a+1}+1 & 0 & \cdots & 0 \\
		(-1)^a\binom{a}{2} & (-1)^{a+1}\binom{a-1}{1} & (-1)^{a+2}+1 & \cdots & 0 \\
		\vdots & \vdots & \vdots & \ddots & \vdots \\
		(-1)^a\binom{a}{b-1} & (-1)^{a+1}\binom{a-1}{b-2} & (-1)^{a+2}\binom{a-2}{b-3} & \cdots & (-1)^{a+b-1}+1
	\end{array}\right)\in M_{b}(\mathbb{F}_{p^m}) .
\end{eqnarray}

It is clear that the matrices \( M(a) \), \( N(a, t) \), and \( K(a, t) \) can be expressed in terms of \( T(a+b, b) \) as follows:
\[
M(a) = T(2^s, 2^s - a), \quad N(a, t) = T(a+2^{s-1} - 2t, 2^{s-1} - t), \quad \text{and} \quad K(a, t) = T(2^s - 2t, 2^s - a - t).
\]

\begin{lemma}\cite[Proposition 3.3]{Kiah2012}\label{hjZHDZZZ}
	For \(p = 2\), the nullity \(\kappa\) of the matrix \(T(a+b, b)\) is given by \(\kappa = \lceil \frac{b+1}{2} \rceil\).
\end{lemma}
\begin{theorem}\label{seeeceed}
	Let \( N \) denote the number of self-dual cyclic codes given in Theorem~\ref{11111}. Then,
	\[
	N = \sum_{a=2^{s-1}}^{2^s-1} \left(2^m\right)^{\lceil \frac{2^s - a + 1}{2} \rceil}.
	\]
\end{theorem}

\begin{proof} 
	For \( 2^{s-1} \leq a \leq 2^s - 1 \), we have \( M(a) = T(2^s, 2^s - a) \). By Lemma~\ref{hjZHDZZZ}, its nullity is given by \( \lceil \frac{2^s - a + 1}{2} \rceil \). Consequently, the number of solutions to Eq.~\eqref{bbbb} is \( \left(2^m\right)^{\lceil \frac{2^s - a + 1}{2} \rceil} \). By summing over all \( a \), the result follows.
\end{proof}

Assuming \( b \geq 2 \), we consider the equations:
\begin{eqnarray}\label{eq:main_equation}
	T(a+b,b)(x_0, x_1, \ldots, x_{b-1})^{\mathrm{tr}} = (0, y_1, \ldots, y_{b-1})^{\mathrm{tr}},
\end{eqnarray}
and
\begin{eqnarray}\label{eq:reduced_equation}
	T(a+b,b)(0, x_1, \ldots, x_{b-1})^{\mathrm{tr}} = (0, y_1, \ldots, y_{b-1})^{\mathrm{tr}},
\end{eqnarray}
where \( x_i, y_i \in \mathbb{F}_{2^m} \). The number of solutions to Eq.~\eqref{eq:main_equation} with \( x_0 \neq 0 \) is given by \( n_1 - n_2 \), where \( n_1 \) and \( n_2 \) are the numbers of solutions to Eqs.~\eqref{eq:main_equation} and \eqref{eq:reduced_equation}, respectively. Moreover, Eq.~\eqref{eq:reduced_equation} is equivalent to:
\begin{eqnarray}\label{eq:reduced_equation2}
	T(a+b-1,b-1)(x_1, \ldots, x_{b-1})^{\mathrm{tr}} = (y_1, \ldots, y_{b-1})^{\mathrm{tr}}.
\end{eqnarray}

By Lemma~\ref{hjZHDZZZ}, the values of \( n_1 \) and \( n_2 \) are given by:
\[
n_1 = 
\begin{cases} 
	\left( 2^m \right)^{\lceil \frac{b+1}{2} \rceil}, & \text{if a solution to Eq.~\eqref{eq:main_equation} exists,} \\[5pt]
	0, & \text{otherwise.}
\end{cases}
\]
\[
n_2 = 
\begin{cases} 
	\left( 2^m \right)^{\lceil \frac{b}{2} \rceil}, & \text{if a solution to Eq.~\eqref{eq:reduced_equation2} exists,} \\[5pt]
	0, & \text{otherwise.}
\end{cases}
\]

With the notations of Theorem~\ref{22222}, if \( t_1 = 2^{s-1} - 1 \), then \( N(a, t_1) \) is a \( 1 \times 1 \) null matrix. Since \( a_0 \) is a unit in \( \mathbb{F}_{2^m} \), the number of solutions to Eq.~\eqref{ccczcz} is \( 2^m - 1 \).

 For \( t_1 < 2^{s-1} - 1 \), based on the discussion above, the number of solutions to Eq.~\eqref{ccczcz} is given by:
\[
\left( 2^m \right)^{\lceil \frac{2^{s-1} - t_1 + 1}{2} \rceil} - \left( 2^m \right)^{\lceil \frac{2^{s-1} - t_1}{2} \rceil} = 
\begin{cases}
	\left( 2^m - 1 \right) \left( 2^m \right)^{\frac{2^{s-1} - t_1}{2}}, & \text{if } 2^{s-1} - t_1 \text{ is even}, \\[5pt]
	0, & \text{if } 2^{s-1} - t_1 \text{ is odd}.
\end{cases}
\]
Thus, solutions to Eq.~\eqref{ccczcz} exist only when \( 2^{s-1} - t_1 \) is even, in which case their number is \( \left( 2^m - 1 \right) \left( 2^m \right)^{\frac{2^{s-1} - t_1}{2}} \). Now, consider Eq.~\eqref{ccczcz2}. First, note that since \( 2t_1 > a + t_2 \), it follows that \( t_2 < 2^s - a - 1 \). Next, for each solution \( (a_0, a_1, \ldots, a_{2^{s-1} - t_1 - 1}) \) of Eq.~\eqref{ccczcz}, we associate a positive integer \( k \) as follows:
\[
(a_0, a_1, \ldots, a_{2^{s-1} - t_1 - 1}) \longmapsto k,
\]
where \( k \) ranges from \( 1 \) to the total number of solutions to Eq.~\eqref{ccczcz}. We then define:

\[
\delta_{a,t_1,t_2,k} = \begin{cases} 
	1, & \text{if a solution exists for } T(2^s - 2t_2, 2^s - a - t_2)(b_0, b_1, \ldots, b_{2^s - a - t_2 - 1})^{\mathrm{tr}} \\ 
	& \quad = (c_0, c_1, \ldots, c_{2^s - a - t_2 - 1})^{\mathrm{tr}}, \\[5pt]
	0, & \text{otherwise,}
\end{cases}
\]
and
\[
\delta_{a,t_1,t_2,k}^\prime = \begin{cases} 
	1, & \text{if a solution exists for } T(2^s - 2t_2 - 1, 2^s - a - t_2 - 1)(b_1, \ldots, b_{2^s - a - t_2 - 1})^{\mathrm{tr}} \\ 
	& \quad = (c_1, \ldots, c_{2^s - a - t_2 - 1})^{\mathrm{tr}}, \\[5pt]
	0, & \text{otherwise.}
\end{cases}
\]

The number of solutions to Eq.~\eqref{ccczcz2} is given by:
\[
\tau_{a,t_1,t_2,k} = \left( 2^m \right)^{\lceil \frac{2^s - a - t_2 + 1}{2} \rceil} \delta_{a,t_1,t_2,k} - \left( 2^m \right)^{\lceil \frac{2^s - a - t_2}{2} \rceil} \delta_{a,t_1,t_2,k}^\prime.
\]

We summarize this in the following theorem:

\begin{theorem}
	With the above notation, let \( N^\prime \) denote the number of self-dual cyclic codes given in Theorem~\ref{22222}. Then,
	\[
	N^\prime =	\sum\limits_{ a=2^{s-1}}^{ 2^{s} -1 }
	\sum\limits_{t_2=0}^{2^s - a-2} 
	\sum\limits_{k=1}^{2^m-1} 
	\tau_{a, 2^{s-1}-1, t_2,k}+
	\sum\limits_{\substack{0 \leq t_1 \leq 2^{s-1} - 2 \\ 2^{s-1} \leq a \leq 2^{s-1} + t_1 \\ t_1 \text{ and } a \text{ have the same parity}}}
	\sum\limits_{t_2=0}^{2t_1 - a} 
	\sum\limits_{k=1}^{ \left( 2^m - 1 \right) \left( 2^m \right)^{\frac{2^{s-1} - t_1}{2}} } 
	\tau_{a, t_1, t_2,k}.
	\]
\end{theorem}

 In \cite[Theorem 4.1]{kim2020classification}, a classification of self-dual cyclic codes of length \( 2^m \) over \(\frac{\mathbb{Z}_2[u]}{\langle u^3 \rangle}\) is provided. However, some cases are missing and not fully considered. 

For instance, when \( h_1(x) = 0 \) and \( a = 2^{s-1} \), \cite[Theorem 4.1(a.2)]{kim2020classification} assumes \( k_2 = 0 \), implying that \( h(x) \) must be a unit. Similarly, for \( h_1(x) = 0 \) and \( a > 2^{s-1} \), \cite[Theorem 4.1(b.2)]{kim2020classification} assumes \( k_2 > 0 \). However, these conditions are not necessary for self-duality. Moreover, the self-dual cyclic code of type \( \left\langle u(x+1)^{2^{s-1}}, u^2 \right\rangle \) is not considered in \cite[Theorem 4.1(a.2)]{kim2020classification}. 

To illustrate this, we revisit \cite[Example 1]{kim2020classification} and observe the existence of additional self-dual cyclic codes that are not identified in the classification.

\begin{example}
	\begin{itemize}
		\item \textbf{Case 1:} For \( h_1(x) = 0 \), Theorem~\ref{seeeceed} shows that the total number of self-dual cyclic codes is:
		\[
		N = \sum\limits_{a=4}^{7} 2^{\lceil \frac{8 - a + 1}{2} \rceil} = 18.
		\]
		
		\item \textbf{Case 2:} When \( h_1(x) \) is a unit, the possible values of the triplet \( (a, t_1, t_2) \) are:
		\[
		(4, 3, 0), \quad (4, 3, 1), \quad (5, 3, 0).
		\]
		In each case, \( N(a, t_1) = T(a + 2^{s-1} - t_1, 2^{s-1} - t_1) = 0 \), and thus \( h_1(x) = 1 \).
		
		\begin{itemize}
			\item When \( (a, t_1, t_2) = (4, 3, 0) \):
			\[
			(c_0, c_1, c_2, c_3) = (0, 0, 1, 1), \quad T(2^s - t_2, 2^s - a - t_2) = 
			\begin{pmatrix}
				0 & 0 & 0 & 0 \\
				0 & 0 & 0 & 0 \\
				0 & 1 & 0 & 0 \\
				0 & 1 & 0 & 0
			\end{pmatrix}.
			\]
			Then, \( (b_0, b_1, b_2, b_3) \) can be:
			\[
			(1, 1, 0, 0), \quad (1, 1, 0, 1), \quad (1, 1, 1, 0), \quad (1, 1, 1, 1).
			\]
			This implies:
			\[
			h_2(x) = 1 + (x+1), \quad h_2(x) = 1 + (x+1) + (x+1)^3, \quad h_2(x) = 1 + (x+1) + (x+1)^2, \quad h_2(x) = 1 + (x+1) + (x+1)^2 + (x+1)^3.
			\]
			
			\item When \( (a, t_1, t_2) = (4, 3, 1) \):
			\[
			(c_0, c_1, c_2) = (0, 1, 1), \quad T(2^s - t_2, 2^s - a - t_2) = 
			\begin{pmatrix}
				0 & 0 & 0 \\
				0 & 0 & 0 \\
				0 & 1 & 0
			\end{pmatrix}.
			\]
			In this case, \( \delta(a, t_1, t_2) = 0 \), so no self-dual code exists for these values. Consequently, \( (b_0, b_1, b_2) \) and \( h_2(x) \) do not exist.
			
			\item When \( (a, t_1, t_2) = (5, 3, 0) \):
			\[
			(c_0, c_1, c_2) = (0, 1, 0), \quad T(2^s - t_2, 2^s - a - t_2) = 
			\begin{pmatrix}
				0 & 0 & 0 \\
				1 & 0 & 0 \\
				0 & 0 & 0
			\end{pmatrix}.
			\]
			Then, \( (b_0, b_1, b_2) \) can be:
			\[
			(1, 0, 0), \quad (1, 0, 1), \quad (1, 1, 0), \quad (1, 1, 1).
			\]
			This implies:
			\[
			h_2(x) = 1, \quad h_2(x) = 1 + (x+1), \quad h_2(x) = 1 + (x+1) + (x+1)^2, \quad h_2(x) = 1 + (x+1)^2.
			\]
		\end{itemize}
	\end{itemize}
\begin{table}[h!]
	\centering
	\caption{Self-dual \(\gamma\)-constacyclic codes of length \(p^s\) over \(R\).}
	\begin{tabular}{|c|c|l|}
		\hline
		$ 	h_1(x) $ & Type & The code $C$ \\ \hline
		
		\multirow{1}{*}{ --}
		& Type 4 & \( \left\langle u(x+1)^{4}, u^2 \right\rangle \) \\ \hline
		
		\multirow{18}{*}{$ 	h_1(x)=0 $} 
		& \multirow{4}{*}{ Type 5}
		& \( \left\langle (x+1)^4  \right\rangle \) \\
		&	& \(\left\langle  (x+1)^4+u^2 \right\rangle\) \\
		&	& \(\left\langle  (x+1)^4+u^2(x+1)^2 \right\rangle\) \\
		&	& \(\left\langle  (x+1)^4+u^2(x+1)^3 \right\rangle\) \\
		&	& \( \left\langle (x+1)^4 + u^2(1 + (x+1)^2)\right\rangle \) \\
		&	& \( \left\langle (x+1)^4 + u^2\left( 1 + (x+1)^3\right) \right\rangle \) \\
		&	& \( \left\langle (x+1)^4 + u^2(x+1)^2(1 + (x+1))\right\rangle \) \\
		&& \( \left\langle (x+1)^4 + u^2(1 + (x+1) + (x+1)^3) \right\rangle\) \\
		\cline{2-3}
		& \multirow{2}{*}{Type 7}	& \(\left\langle  (x+1)^6+u^2,u(x+1)^4,u^2 (x+1)^2\right\rangle\) \\&	& \(\left\langle  (x+1)^7+u^2,u(x+1)^4,u^2 (x+1)\right\rangle\) \\
		\cline{2-3}
		& \multirow{8}{*}{Type 8}
		& \(\left\langle  (x+1)^5,u(x+1)^4,u^2 (x+1)^3\right\rangle\) \\
		&	& \(\left\langle  (x+1)^5+u^2(x+1),u(x+1)^4,u^2 (x+1)^3\right\rangle\) \\
		&	& \(\left\langle  (x+1)^5+u^2(x+1)^2,u(x+1)^4,u^2 (x+1)^3\right\rangle\) \\
		&	& \(\left\langle  (x+1)^5+u^2(x+1)\left( 1+(x+1)\right) ,u(x+1)^4,u^2 (x+1)^3\right\rangle\) \\
		&	& \(\left\langle  (x+1)^6,u(x+1)^4,u^2 (x+1)^2\right\rangle\) \\
		
		&	& \(\left\langle  (x+1)^5+u^2(x+1),u(x+1)^4,u^2 (x+1)^2\right\rangle\) \\
		&	& \(\left\langle  (x+1)^6+u^2\left( 1+(x+1)\right) ,u(x+1)^4,u^2 (x+1)^2\right\rangle\) \\

		&	& \(\left\langle  (x+1)^7,u(x+1)^4,u^2 (x+1)\right\rangle\) \\

		\hline
		
		\multirow{8}{*}{$ h_1(x) $ is unit} 
		& \multirow{8}{*}{Type 7}
		& \(\left\langle  (x+1)^4+u(x+1)^3+u^2\left( 1+(x+1) \right) \right\rangle\) \\
		&	& \(\left\langle  (x+1)^4+u(x+1)^3+u^2\left( 1+(x+1)+(x+1)^2 \right) \right\rangle\) \\
		&	& \(\left\langle  (x+1)^4+u(x+1)^3+u^2\left( 1+(x+1)+(x+1)^3 \right) \right\rangle\) \\
		&	& \(\left\langle  (x+1)^4+u(x+1)^3+u^2\left( 1+(x+1)+(x+1)^2+(x+1)^3 \right) \right\rangle\) \\
		&	& \(\left\langle  (x+1)^5+u(x+1)^3+u^2,(x+1)^4+u^2(x+1)^2,u^2(x+1)^3\right\rangle\) \\
		&	& \(\left\langle  (x+1)^5+u(x+1)^3+u^2\left( 1+(1+x)\right),(x+1)^4+u^2(x+1)^2,u^2(x+1)^3\right\rangle\) \\
		&	& \(\left\langle  (x+1)^5+u(x+1)^3+u^2\left( 1+(1+x)^2\right),(x+1)^4+u^2(x+1)^2,u^2(x+1)^3\right\rangle\) \\
		&	& \(\left\langle  (x+1)^5+u(x+1)^3+u^2\left( 1+(1+x)+(1+x)^2\right),(x+1)^4+u^2(x+1)^2,u^2(x+1)^3\right\rangle\)\\\hline 
		
	\end{tabular}
	\label{tab:self_dual_codes}
\end{table}
\end{example}


\begin{thebibliography}{99}
		\bibitem{hammons1994z}
		A. R. Hammons, P. V. Kumar, A. R. Calderbank, N. J. A. Sloane, P. Solé, The $\mathbb{Z}_4$-linearity of Kerdock, Preparata, Goethals, and related codes, \emph{IEEE Transactions on Information Theory}, vol. 40, no. 2, pp. 301--319, 1994.
		
		\bibitem{udaya1998optimal} 
	P. Udaya and M. U. Siddiqi,  
	"Optimal large linear complexity frequency hopping patterns derived from polynomial residue class rings,"  
	\emph{IEEE Transactions on Information Theory}, vol. 44, no. 4, pp. 1492--1503, 1998.
		
		\bibitem{han2011cyclic} 
	M. Han, Y. Ye, S. Zhu, C. Xu, and B. Dou,  
	"Cyclic codes over \( R = \mathbb{F}_p + u\mathbb{F}_p + \cdots + u^{k-1}\mathbb{F}_p \) with length \( p^s n \),"  
	\emph{Information Sciences}, vol. 181, no. 4, pp. 926--934, 2011.
		
		
		\bibitem{dinh2010constacyclic}
	H. Q. Dinh,  
	"Constacyclic codes of length \( p^{s} \) over \( \mathbb{F}_{p^{m}} + u\mathbb{F}_{p^{m}} \),"  
	\emph{Journal of Algebra}, vol. 324, no. 5, pp. 940--950, 2010.
		
		\bibitem{Kiah2012}
H. M. Kiah, K. H. Leung, and S. Ling,  
"A note on cyclic codes over \( \mathrm{GR}(p^2, m) \) of length \( p^k \),"  
\emph{Designs, Codes and Cryptography}, vol.~63, pp.~105--112, 2012.
		
		
		
		
		
		
		\bibitem{dinh2020constacyclic}

	H. Q. Dinh, B. T. Nguyen, and W. Yamaka,  
	"Constacyclic Codes of Length \( 3p^s \) Over \( \mathbb{F}_{p^m} + u\mathbb{F}_{p^m} \) and Their Application in Various Distance Distributions,"  
	\emph{IEEE Access}, vol. 8, pp. 204031--204056, 2020.
	
	
	\bibitem{dinh2017constacyclic}
	H. Q. Dinh, S. Dhompongsa, and S. Sriboonchitta,  
	"On constacyclic codes of length \( 4p^s \) over \( \mathbb{F}_{p^m} + u\mathbb{F}_{p^m} \),"  
	\emph{Discrete Mathematics}, vol. 340, pp. 832--849, 2017.
	
	
		
	
		
		
\bibitem{dinh2004cyclic}
H. Q. Dinh and S. R. López-Permouth,  
"Cyclic and negacyclic codes over finite chain rings,"  
\emph{IEEE Transactions on Information Theory}, vol. 50, no. 8, pp. 1728--1744, Aug. 2004.
		
		
		\bibitem{dinh2018self}
		
	H. Q. Dinh, Y. Fan, H. Liu, X. Liu, and S. Sriboonchitta,  
	"On self-dual constacyclic codes of length \( p^s \) over \( \mathbb{F}_{p^m} + u\mathbb{F}_{p^m} \),"  
	\emph{Discrete Mathematics}, vol. 341, pp. 324--335, 2018.
		
		
		
		\bibitem{huffman2010fundamentals}
	W. C. Huffman and V. Pless,  
	\emph{Fundamentals of Error-Correcting Codes},  
	Cambridge University Press, 2003.
		
		
		\bibitem{lidl1983finite}
	R. Lidl,  
	\emph{Finite Fields},  
	\emph{Encyclopedia of Mathematics and its Applications}, vol. 20, pp. 1--46, 1983.
		
		\bibitem{dinh2009constacyclic} 
	H. Q. Dinh,  
	"Constacyclic Codes of Length \( 2^s \) Over Galois Extension Rings of \( \mathbb{F}_2 + u\mathbb{F}_2 \),"  
	\emph{IEEE Transactions on Information Theory}, vol. 55, no. 4, pp. 1730--1740, 2009.
	

		
		\bibitem{sobhani2015complete} 
	R. Sobhani,  
	"Complete classification of \((\delta + \alpha u^2)\)-constacyclic codes of length \( p^k \) over \( \mathbb{F}_{p^m} + u\mathbb{F}_{p^m} + u^2\mathbb{F}_{p^m} \),"  
	\emph{Finite Fields and Their Applications}, vol. 34, pp. 123--138, 2015.
		\bibitem{dinh2017repeated}
	H. Q. Dinh, H. D. T. Nguyen, S. Sriboonchitta, and T. M. Vo,  
	"Repeated-root constacyclic codes of prime power lengths over finite chain rings,"  
	\emph{Finite Fields and Their Applications}, vol. 43, pp. 22--41, 2017.
		
	
		
		\bibitem{laaouine2021complete}
H. Q. Dinh, H. D. T. Nguyen, S. Sriboonchitta, and T. M. Vo,  
"Repeated-root constacyclic codes of prime power lengths over finite chain rings,"  
\emph{Finite Fields and Their Applications}, vol. 43, pp. 22--41, 2017.
		
		
		
			\bibitem{dinh2016repeated}
	H. Q. Dinh, S. Dhompongsa, and S. Sriboonchitta,  
	"Repeated-root constacyclic codes of prime power length over \( \mathbb{F}_{p^m}[u]/\langle u^a \rangle \) and their duals,"  
	\emph{Discrete Mathematics}, vol. 339, no. 6, pp. 1706--1715, 2016.
		\bibitem{chen2016constacyclic}
		B. Chen, H. Q. Dinh, H. Liu, and L. Wang,  
		"Constacyclic codes of length \( 2p^s \) over \( \mathbb{F}_{p^m} + u\mathbb{F}_{p^m} \),"  
		\emph{Finite Fields and Their Applications}, vol. 37, pp. 108--130, 2016.
		
		
		
		\bibitem{dinh2008linear}
	H. Q. Dinh,  
	"On the linear ordering of some classes of negacyclic and cyclic codes and their distance distributions,"  
	\emph{Finite Fields and Their Applications}, vol. 14, no. 1, pp. 22--40, 2008, Elsevier.
		
		
		
		
		
		
		\bibitem{dinh2019class}
	H. Q. Dinh, B. T. Nguyen, S. Sriboonchitta, and T. M. Vo,  
	"On a class of constacyclic codes of length \( 4p^s \) over \( \mathbb{F}_{p^m} + u\mathbb{F}_{p^m} \),"  
	\emph{Journal of Algebra and Its Applications}, vol. 18, no. 02, pp. 1950022, 2019.
		
		\bibitem{dinh2018negacyclic}
	H. Q. Dinh, B. T. Nguyen, and S. Sriboonchitta,  
	"Negacyclic codes of length \( 4p^s \) over \( \mathbb{F}_{p^m} + u\mathbb{F}_{p^m} \) and their duals,"  
	\emph{Discrete Mathematics}, vol. 341, no. 4, pp. 1055--1071, 2018, Elsevier.
	
	
	
		\bibitem{dinh2018cyclic}
H. Q. Dinh, A. Sharma, S. Rani, and S. Sriboonchitta,  
"Cyclic and negacyclic codes of length \( 4p^s \) over \( \mathbb{F}_{p^m} + u\mathbb{F}_{p^m} \),"  
\emph{Journal of Algebra and Its Applications}, vol. 17, no. 09, p. 1850173, 2018, World Scientific.
		
	\bibitem{dinh2019alpha+}
	H. Q. Dinh, B. T. Nguyen, S. Sriboonchitta, and T. M. Vo,  
	"On \((\alpha + u \beta)\)-constacyclic codes of length \( 4p^s \) over \( \mathbb{F}_{p^m} + u\mathbb{F}_{p^m}^* \),"  
	\emph{Journal of Algebra and Its Applications}, vol. 18, no. 02, p. 1950023, 2019, World Scientific.
	
	\bibitem{ahendouz122negacyclic}
	Y. Ahendouz and I. Akharraz,  
	"On Negacyclic Codes of Length \( 8p^s \) over \( \mathbb{F}_{p^m} + u\mathbb{F}_{p^m} \),"  
	\emph{Journal of Combinatorial Mathematics and Combinatorial Computing}, vol. 122, pp. 351--359.
	
	\bibitem{ahendouz160note}
	Y. Ahendouz and I. Akharraz,  
	"A Note on Complete Classification of Repeated-Root \(\sigma\)-Constacyclic Codes of Prime Power Length over \(\mathbb{F}_{p^m}[u]/\langle u^3 \rangle\) and Their Hamming Distances,"  
	\emph{Ars Combinatoria}, vol. 160, pp. 105--115.
	
	
	\bibitem{erfanian2024duality}
	A. Erfanian and R. M. Hesari,  
	"On the duality of cyclic codes of length \( p^s \) over \( \mathbb{F}_{p^m}[u] \),"  
	\emph{Finite Fields and Their Applications}, vol. 99, p. 102500, 2024.
	
	\bibitem{kim2020classification}
	B. Kim and Y. Lee,  
	"Classification of self-dual cyclic codes over the chain ring \(\mathbb{Z}_p[u]/\langle u^3 \rangle\),"  
	\emph{Designs, Codes and Cryptography}, vol. 88, no. 10, pp. 2247--2273, 2020.
		
		
	\end{thebibliography}
\end{document}